\documentclass[11pt]{article}

\usepackage[T1]{fontenc}
\usepackage{lmodern}
\usepackage{microtype}
\usepackage{authblk}
\usepackage[margin=1.08in]{geometry}
\usepackage{amsmath,amssymb,amsthm,mathtools}
\usepackage{booktabs}
\usepackage{tabularx}
\usepackage{ragged2e}
\usepackage{enumitem}
\usepackage{xcolor}
\usepackage{hyperref}
\usepackage[capitalise,noabbrev]{cleveref}
\hypersetup{
	colorlinks=true,
	linkcolor=blue!55!black,
	citecolor=blue!55!black,
	urlcolor=blue!55!black,
	pdfauthor={Daqing Wan, Jun Zhang},
	pdftitle={NP-Hardness of Bounded Distance Decoding for Reed-Solomon Codes}
}

\allowdisplaybreaks
\setlist[itemize]{leftmargin=1.6em,itemsep=0.25em,topsep=0.35em}
\setlist[enumerate]{leftmargin=1.9em,itemsep=0.25em,topsep=0.35em}

\newtheorem{theorem}{Theorem}[section]
\newtheorem{lemma}[theorem]{Lemma}
\newtheorem{proposition}[theorem]{Proposition}

\theoremstyle{definition}
\newtheorem{definition}[theorem]{Definition}

\theoremstyle{remark}
\newtheorem{remark}[theorem]{Remark}

\newcommand{\F}{\mathbb{F}}
\newcommand{\Kfield}{\mathbb{K}}

\newcommand{\D}{\mathcal{D}}
\newcommand{\C}{\mathcal{C}}
\newcommand{\MSS}{\operatorname{MSS}}
\newcommand{\RS}{\operatorname{RS}}

\DeclareRobustCommand{\RSBDD}{\textup{RS-BDD}}
\newcommand{\Span}{\operatorname{span}}
\newcommand{\dist}{\operatorname{dist}}

\newcommand{\poly}{\operatorname{poly}}
\newcommand{\eps}{\varepsilon}
\newcommand{\one}{\mathbf{1}}
\newcolumntype{Y}{>{\RaggedRight\arraybackslash}X}

\title{\textbf{NP-Hardness of Bounded Distance Decoding\\[0.18em]
		for Reed--Solomon Codes}\thanks{AI assistance disclosure: The arithmetic input is due to the authors. OpenAI\textquoteright s ChatGPT Pro was used to assist with the development and exposition of the reduction, drafting and revision of the manuscript, and checks of mathematical consistency and references. This assistance included proposed proof steps and editorial suggestions. The authors are responsible for the mathematical content, the accuracy of the references, and the final manuscript.}}
\author{Daqing Wan and Jun Zhang}

\date{}

\begin{document}
	\maketitle
	
	\begin{abstract}
		For an $[n,K]$ Reed--Solomon code, the covering radius is $n-K$.  Gandikota, Ghazi, and Grigorescu proved deterministic NP-hardness of bounded-distance decoding when the decoding radius is $d$ below the covering radius for every $1\le d\le c\log n/\log\log n$, where $c>0$ is an absolute constant.  We prove that, for every fixed rational $0<\alpha<1/2$, bounded-distance decoding is NP-complete under deterministic polynomial-time many-one reductions over explicitly represented finite extension fields for the additive gap $d=\lfloor n^\alpha\rfloor$ below the covering radius.  The hard codes have odd block length~$n$, dimension
		$K=(n+1)/2-d$, decoding radius $(n-1)/2$, and rate tending to~$1/2$. The alphabet size is subexponential in the evaluation set size: for a fixed $0<\eta<1$ depending only on~$\alpha$, it is $2^{\Theta(n^\eta\log n)}=2^{o(n)}$.
		
		The proof passes through moments subset sum on $n-1$ nonzero field elements, with required subset size $(n-1)/2$ and $d$ prescribed moments.  The arithmetic ingredient is a uniform positive-completion theorem over prime fields $\F_q$ with $q\ge d^{2+\rho}$, for any fixed $\rho>0$.  A sharper form follows from a higher-dimensional point-count estimate based on Deligne's theorem; the weaker form used in our reduction is proved more elementarily using additive-character orthogonality, the one-variable Weil bound, a moment identity of order~$2d$, and Newton identities.  A universal completion pool, an extension-field quotient construction, and a deterministic linear-size simultaneous power condenser complete the reduction.
	\end{abstract}

	\section{Introduction}
	
	Reed--Solomon codes are among the most classical algebraic error-correcting codes~\cite{ReedSolomon1960}.  For distinct evaluation points
	$\boldsymbol{a}=(a_1,\ldots,a_n)$ in a finite field~$\F$, the Reed--Solomon code of dimension~$K$ is
	\begin{equation}
	\RS_K(\boldsymbol{a})
	=
	\left\{
	\bigl(f(a_1),\ldots,f(a_n)\bigr)
	: f\in\F[T],\ \deg f<K
	\right\}.
	\label{eq:rs-definition}
	\end{equation}
	It is an $[n,K,n-K+1]$ maximum-distance-separable code, and its covering radius is $n-K$.  Classical unique decoding corrects up to $\lfloor(n-K)/2\rfloor$ errors, while Guruswami--Sudan list decoding reaches the Johnson radius
	\begin{equation}
	n-\sqrt{n(K-1)}
	\end{equation}
	up to the usual strict inequality and integral rounding~\cite{GuruswamiSudan1999}.  At constant rate $R=K/n$, this gives relative radius approximately $1-\sqrt R$, whereas the covering radius is $(1-R)n$.  Understanding the computational complexity between these two scales has been a longstanding problem.
	
	The field and the evaluation set matter to this question.  Ben-Sasson, Kopparty, and Radhakrishnan exhibited full-length Reed--Solomon codes with superpolynomially large lists beyond the Johnson bound in certain extension-field regimes~\cite{BenSassonKoppartyRadhakrishnan2010}.  Such examples obstruct enumeration of the entire list, but do not by themselves imply hardness of finding one nearby codeword or deciding whether one exists.  In the other direction, recent work shows that randomly punctured Reed--Solomon codes attain combinatorial list-decoding capacity even over linear-sized alphabets~\cite{AlrabiahEtAl2025}; a bound on list size alone does not provide a decoding algorithm.  Brakensiek, Chen, Putterman, Zhang, and Zheng~\cite{Br+26} have now given a deterministic polynomial-time list-decoding algorithm beyond the Johnson radius over prime fields in the low constant-rate regime.  More precisely, for fixed $\theta\in(0,1)$ and sufficiently small constant $\eps>0$, their algorithm corrects a $1-\eps$ fraction of errors when $R\le(1-\theta)\eps$, for arbitrary evaluation sets over suitable prime fields of size $O(n)$, including full-length codes for sufficiently large lengths.  This result does not address exact nearest-codeword decoding or a gap below the covering radius that vanishes as a fraction of~$n$.
	
	Throughout this paper, \emph{bounded-distance decoding} means the following threshold decision problem, without a unique-decoding promise: given a received word~$y$ and a radius~$\tau\le n-K$, decide whether
	\begin{equation}
	\dist\bigl(y,\RS_K(\boldsymbol a)\bigr)\le \tau.
	\end{equation}
	It is convenient to write
	\begin{equation}
	\tau=(n-K)-d,
	\label{eq:intro-gap-parameter}
	\end{equation}
	so that~$d$ measures the gap below the covering radius.  This decision problem is distinct from search decoding, which asks for one codeword in the ball, and list decoding, which asks for all of them.  At $d=0$ the decision problem is trivial and a codeword can be found by interpolation.  At $d=1$, however, Guruswami and Vardy proved NP-hardness~\cite{GuruswamiVardy2005}; equivalently, non-deep-hole recognition is NP-complete and its complement, deep-hole recognition, is coNP-complete.  Cheng and Murray subsequently studied the deep-hole problem and its connection with subset sum~\cite{ChengMurray2007}.
	
	These NP-hardness results allow the evaluation points to be part of the input, and their constructions use only a small subset of the ambient field.  They do not establish NP-hardness for the standard primitive code, whose evaluation set is $\F_Q^\times$, or for the extended full-length code, whose evaluation set is $\F_Q$.  Here $Q$ denotes the alphabet size.  The distinction is substantive: prescribing the entire evaluation set removes the freedom to encode a combinatorial instance in its choice, and filling out a sparse hard instance can require exponentially many coordinates.
	
	Cheng and Wan obtained a complementary hardness result for these prescribed full-length evaluation sets~\cite{ChengWan2007}.  For suitable integers $h$ and~$g$, they gave a randomized oracle reduction from discrete logarithms in $\F_{Q^h}^{\times}$ to search bounded-distance decoding of the $[Q,g-h]_Q$ Reed--Solomon code at radius $Q-g$.  A concrete specialization takes $Q>(h+2)^4$, $g=4h+4$, and hence dimension $3h+4$ and radius $Q-4h-4$.  In our gap notation, $d=h$.  Their factorization theorem represents each nonzero extension-field element as a product of exactly $g$ distinct factors $\xi-a$, with $a\in\F_Q$ and $\F_Q(\xi)=\F_{Q^h}$.  A decoding oracle supplies such factorizations, which yield the multiplicative relations needed for index calculus.  The reduction has overhead polynomial in~$Q$; choosing $Q$ polynomially bounded in~$h$ also makes this polynomial in the discrete-logarithm input length $h\log Q$.
	
	The initial construction has vanishing rate.  Cheng and Wan later removed this restriction, proving discrete-logarithm hardness of maximum-likelihood decoding for explicit full-length families of any prescribed limiting rate in $(0,1)$~\cite{ChengWan2010}.  Their positive-rate construction uses a proper subfield and additional distinct factors to raise the agreement count.  Its special received words have exact distance $Q-g=(Q-K)-h$, so a search decoder at that radius suffices.  Thus additive gaps growing as a fixed positive power of the block length already occur in this conditional search-hardness framework; the issue addressed here is an additive gap of order $n^\alpha$ under deterministic NP-hardness reductions.\footnote{Augot and Morain developed the decoding--index-calculus connection algorithmically~\cite{AugotMorain2012}.  Franco Garrido and Chailloux revisit it in a recent preprint on quantum decoding reductions, and also discuss NP-hardness at vanishing rates for constructed evaluation sets~\cite{FrancoGarridoChailloux2026}.}
	
	Discrete-logarithm hardness should also be distinguished from NP-hardness and from an unconditional running-time lower bound.  In particular, the quasipolynomial finite-field discrete-logarithm algorithms of Kleinjung and Wesolowski~\cite{KleinjungWesolowski2022} apply in relevant polynomially related parameter regimes.  This does not invalidate the Cheng--Wan reduction, but it limits the quantitative hardness evidence available from it.  Moreover, their decoding instances are promised to have a codeword at the specified radius, so the reduction is a search-hardness statement, rather than an NP-hardness proof for the threshold decision problem studied here.
	
	For deterministic NP-hardness, Gandikota, Ghazi, and Grigorescu first extended the gap from $d=1$ to $d=2,3$~\cite{GandikotaGhaziGrigorescu2015}, and then obtained the first result with a growing gap below the covering radius~\cite{GandikotaGhaziGrigorescu2018}.  They proved that there is an absolute constant $c>0$ such that Reed--Solomon bounded-distance decoding is NP-hard for every integer
	\begin{equation}
	1\le d\le c\frac{\log n}{\log\log n}
	\label{eq:ggg-range}
	\end{equation}
	under deterministic polynomial-time reductions, for constant-rate codes over sufficiently large prime fields with constructed evaluation sets.  They also obtained a gap of order $\log n$ under quasipolynomial-time reductions.  Their work explicitly asked whether one can prove NP-hardness at a substantially larger distance below the covering radius.
	
	Our main theorem gives an additive gap of order $n^\alpha$ below the covering radius for every fixed rational $0<\alpha<1/2$.
	
	\begin{theorem}[Main theorem]
		\label{thm:main-bdd}
		For every fixed rational $0<\alpha<1/2$, bounded-distance decoding is NP-complete under deterministic polynomial-time many-one reductions even when the input code is an $[n,K]$ Reed--Solomon code over an explicitly represented finite extension field with odd block length~$n$ and
		\begin{equation}
		d=\lfloor n^\alpha\rfloor,
		\qquad
		K=\frac{n+1}{2}-d.
		\label{eq:intro-code-parameters}
		\end{equation}
		The decoding radius is
		\begin{equation}
		\tau=\frac{n-1}{2}
		=(n-K)-d.
		\end{equation}
		In particular, the code rate tends to $1/2$, and the gap below the covering radius is exactly $\lfloor n^\alpha\rfloor$. The hard instances can be chosen over fields of size $2^{\Theta(n^\eta\log n)}=2^{o(n)}$, where $0<\eta<1$ is a constant depending only on~$\alpha$.
	\end{theorem}
	
	The alphabet-size bound is a further quantitative feature of the construction. In the growing-gap result of Gandikota, Ghazi, and Grigorescu, the fields have size $2^{\poly(n)}$, described there as exponentially large~\cite[Theorem~1.1 and the discussion following it]{GandikotaGhaziGrigorescu2018}. Here we obtain the explicit subexponential bound
	\begin{equation}
	Q=|\Kfield|=2^{\Theta(n^\eta\log n)}=2^{o(n)},
	\qquad 0<\eta<1,
	\label{eq:intro-alphabet-size}
	\end{equation}
	together with the additive gap of order $n^\alpha$. The calculation appears in \Cref{subsec:size-running-time}. This bound is still superpolynomial in~$n$; NP-hardness over polynomial-size alphabets remains open in this parameter regime.
	
	The evaluation set in \Cref{thm:main-bdd} is constructed by the reduction inside a larger extension field; it is not the full field.  Our result therefore strengthens the deterministic NP-hardness line of work without resolving NP-hardness for standard primitive or extended full-length Reed--Solomon codes.  It also does not improve the prime-field gap in~\eqref{eq:ggg-range}.  The improvement is additive: since $d=o(n)$ and $K/n\to1/2$, the relative hard radius still tends to $1/2$, while the Johnson radius tends to $1-1/\sqrt2$.  These parameter distinctions also separate our theorem from the low-rate prime-field algorithm of~\cite{Br+26}.
	
	The intermediate problem in the proof is \emph{moments subset sum}.  Given a set $A\subseteq\F$, a cardinality~$k$, and targets $m_1,\ldots,m_d$, one asks whether some $k$-element subset $S\subseteq A$ satisfies
	\begin{equation}
	\sum_{x\in S}x^j=m_j,
	\qquad 1\le j\le d.
	\label{eq:mss-intro}
	\end{equation}
	For $d=1$ this is fixed-cardinality subset sum over a finite field.  When the characteristic exceeds~$d$, Newton identities convert the first $d$ power sums into the first $d$ elementary symmetric functions, giving a direct reduction from moments subset sum to Reed--Solomon decoding at a radius exactly $d$ below the covering radius.  Gandikota, Ghazi, and Grigorescu introduced this formulation and proved hardness in the range~\eqref{eq:ggg-range}.
	
	The \emph{Prouhet--Tarry--Escott (PTE) problem} asks for two distinct multisets of integers $X=\{x_1,\ldots,x_s\}$ and $Y=\{y_1,\ldots,y_s\}$ of the same cardinality such that
	\begin{equation}
	\sum_{i=1}^{s}x_i^j=\sum_{i=1}^{s}y_i^j,
	\qquad 1\le j\le d.
	\label{eq:pte-definition}
	\end{equation}
	The finite-field version imposes these identities in a finite field; an inhomogeneous version prescribes their differences instead of requiring them to vanish. The construction in~\cite{GandikotaGhaziGrigorescu2018} uses explicit PTE-type gadgets to match higher moments while preserving the Boolean choices.  Exponential growth of the gadgets in~$d$, together with their bit-size requirements and the need to exclude unintended auxiliary subsets, obstructs an additive gap growing as a fixed positive power of~$n$ in that approach; see \Cref{sec:history}.
	
	We prove the following stronger intermediate result.  To keep the coding parameter uniform, we write $n=|A|+1$ for an MSS universe~$A$; this is exactly the block length produced by the MSS-to-Reed--Solomon reduction.
	
	\begin{theorem}[Polynomially many moments]
		\label{thm:main-mss}
		For every fixed rational $0<\alpha<1/2$, the following restriction of moments subset sum is NP-complete under deterministic polynomial-time many-one reductions:
		\begin{itemize}
			\item the field $\F$ is an explicitly represented finite extension field of characteristic larger than~$d$;
			\item the input set~$A$ has even size $n-1$ and contains no zero;
			\item the required subset size is exactly $(n-1)/2$;
			\item the number of moments is $d=\lfloor n^\alpha\rfloor$.
		\end{itemize}
		The hard instances can be chosen over fields of size $2^{\Theta(n^\eta\log n)}=2^{o(n)}$ for a constant $0<\eta<1$ depending only on~$\alpha$.
	\end{theorem}
	
	The arithmetic input is a prescribed-target completion theorem.  Fix $\rho>0$.  For all sufficiently large~$d$, if $q$ is prime with $q\ge d^{2+\rho}$ and $h=C_\rho d+O(1)$ for a suitable constant $C_{\rho}$, then every vector $\boldsymbol c\in\F_q^d$ can be represented by distinct nonzero elements:
	\begin{equation}
	c_j=\sum_{a\in S}a^j,
	\qquad
	1\le j\le d,
	\qquad
	|S|=h.
	\label{eq:intro-completion}
	\end{equation}
	The higher-dimensional estimate in~\cite{Wan2026} permits any constant
	$C_\rho>1+2/\rho$, whereas the more elementary Weil-bound argument used in our reduction permits the weaker range
	$C_\rho>2+2/\rho$.
	Our initial proof used the higher-dimensional point-count estimate for locally dense Reed--Solomon lattices~\cite{Wan2026}.  That estimate ultimately rests on the Weil--Deligne theory for higher-dimensional varieties.  We later found that the complexity reduction needs only a slightly weaker numerical constant and that this weaker form has a more elementary proof using additive-character orthogonality, the classical one-variable Weil bound, an exact moment identity of order~$2d$, and Newton identities~\cite{Schmidt1976}.  We retain the sharper statement for comparison, but the proof of \Cref{thm:main-bdd} uses only the more elementary version.
	
	The completion in~\eqref{eq:intro-completion} is existential, which is enough.  The reduction places a polynomial-size universal pool of possible correcting elements in the instance, and the NP witness chooses a suitable completion.  Soundness is protected independently by an additive quotient.  We work over an extension field $\Kfield/\F_q$ and choose a proper $\F_q$-subspace $H\subset\Kfield$ satisfying
	\begin{equation}
	\Span_{\F_q}(H\cdot H)=\Kfield.
	\label{eq:quotient-properties}
	\end{equation}
	The Boolean source instance is stored in $\Kfield/H$, while all auxiliary elements lie in~$H$.  Hence every auxiliary subset disappears from the first moment after quotienting, whereas higher powers of directions in~$H$ can span the full extension field.  A deterministic simultaneous power condenser supplies only $O([\Kfield:\F_q])$ such directions.
	
	The extension-field hypothesis is essential to this particular soundness mechanism.  Viewed as a vector space over itself, a prime field has no nonzero proper additive subspace: taking $H=0$ leaves no auxiliary pool, while taking $H=\F_q$ destroys the quotient information.  Thus our argument does not improve the prime-field result of Gandikota, Ghazi, and Grigorescu; a different soundness mechanism would be required.
	
	Two further challenges remain: preserving NP-hardness with a full-length evaluation set, and obtaining a gap $d=\Omega(n)$ at constant rate.  Even the Cheng--Wan representation mechanism has a counting limitation for the latter goal.  If every element of $\F_{Q^h}^{\times}$ is represented by a product indexed by a $g$-element subset of~$\F_Q$, then necessarily
	\begin{equation}
	Q^h-1\le\binom{Q}{g}\le2^Q,
	\qquad h=O(Q/\log Q).
	\label{eq:intro-cw-counting}
	\end{equation}
	Since its gap is $d=h$ and its length is $Q$, improving the character-sum estimates alone cannot make this gap a positive fraction of the length.  Our completion method encounters a related subset-counting constraint, discussed in \Cref{sec:perspective}, as well as the requirement that every auxiliary subset preserve soundness.  Locating a sharper computational boundary will require attention to these structural constraints and to the separate roles of alphabet size, rate, and list size.
	
	\paragraph{Organization.}
	\Cref{sec:preliminaries} develops the coding-theoretic background and proves the moments-subset-sum to Reed--Solomon BDD reduction.  \Cref{sec:history} reviews the earlier BDD result and its Prouhet--Tarry--Escott bottleneck.  \Cref{sec:completion} proves the elementary completion theorem and records the sharper form.  \Cref{sec:extension-algebra} develops the extension-field quotient construction and deterministic simultaneous power condenser.  \Cref{sec:universal-pool} builds the balanced universal completion pool.  \Cref{sec:hardness} proves the intermediate moments-subset-sum theorem, and \Cref{sec:rs-consequences} completes the proof of the main Reed--Solomon BDD theorem.  \Cref{sec:prime-field} explains the extension-field requirement, and \Cref{sec:perspective} compares the coding-theoretic consequences with earlier work and records open problems.
	
	\paragraph{Acknowledgments.}
	The authors would like to thank Qi Cheng for his careful reading and constructive comments on an earlier version of this paper.
	
	\section{Coding and algebraic preliminaries}
	\label{sec:preliminaries}
	
	\subsection{Reed--Solomon codes, BDD, and deep holes}
	
	For a code $\C\subseteq\F^n$ and a word $y\in\F^n$, write
	\begin{equation}
	\dist(y,\C)=\min_{c\in\C}\Delta(y,c),
	\end{equation}
	where $\Delta$ denotes Hamming distance.  Interpolating a polynomial of degree less than~$K$ through any $K$ positions shows
	\begin{equation}
	\dist\bigl(y,\RS_K(\boldsymbol{a})\bigr)\le n-K.
	\end{equation}
	This bound is attained.  Indeed, if $g$ has degree exactly~$K$, then the received word
	$y_i=g(a_i)$ agrees with a codeword arising from a polynomial $f$ of degree less than~$K$ in at most $K$ positions, because $g-f$ has degree~$K$.  Hence the covering radius is exactly $n-K$.  A word is a deep hole if equality holds.
	
	Following~\cite{GandikotaGhaziGrigorescu2018}, it is convenient to parameterize decoding by the gap below the covering radius.
	
	\begin{definition}[\RSBDD{} with parameter $d$]
		An instance consists of an $[n,K]$ Reed--Solomon code and a received word~$y$.  The question is whether
		\begin{equation}
		\dist(y,\RS_K)\le (n-K)-d.
		\label{eq:rs-bdd-definition}
		\end{equation}
		Equivalently, does a polynomial of degree less than~$K$ agree with~$y$ in at least $K+d$ evaluation positions?
	\end{definition}
	
	At $d=1$, a no-instance is exactly a deep hole, while a yes-instance is a non-deep-hole.  The reduction of Guruswami--Vardy and Cheng--Murray therefore gives NP-completeness of non-deep-hole recognition and coNP-completeness of deep-hole recognition for Reed--Solomon codes~\cite{GuruswamiVardy2005,ChengMurray2007}.
	
	\subsection{Moments subset sum and finite-field PTE}
	
	\begin{definition}[Moments subset sum]
		Let $\F$ be a finite field.  An instance of $\MSS(d)$ consists of a finite set
		$A\subseteq\F$, an integer $k\le |A|$, and targets
		$m_1,\ldots,m_d\in\F$.  The question is whether there is a subset
		$S\subseteq A$ with $|S|=k$ and
		\begin{equation}
		\sum_{x\in S}x^j=m_j,
		\qquad 1\le j\le d.
		\label{eq:mss-definition}
		\end{equation}
	\end{definition}
	
	When $d=1$, this is the fixed-cardinality finite-field subset-sum problem.  For larger~$d$, two subsets $S,T$ in the same moment fiber satisfy
	\begin{equation}
	\sum_{x\in S}x^j=\sum_{y\in T}y^j,
	\qquad 1\le j\le d.
	\label{eq:pte-fiber}
	\end{equation}
	Canceling $S\cap T$ produces disjoint finite-field PTE sets.  Prescribing a nonzero difference of moment vectors gives an inhomogeneous PTE system.
	
	The problem belongs to NP for explicitly represented finite fields: a subset is a certificate, and all moment equations can be checked with polynomially many field operations.
	
	\subsection{A self-contained reduction from MSS to \RSBDD{}}
	\label{subsec:mss-to-rs}
	
	We recall and prove the reduction of~\cite{GandikotaGhaziGrigorescu2018}.  The formulation below is sufficient for our instances, whose elements are all nonzero.
	
	\begin{lemma}[Newton conversion]
		\label{lem:newton-conversion}
		Suppose $\operatorname{char}(\F)>d$.  From targets
		$m_1,\ldots,m_d$ one can compute, in polynomial time, targets
		$E_1,\ldots,E_d$ such that a finite set~$S$ satisfies
		\eqref{eq:mss-definition} if and only if its first $d$ elementary symmetric functions equal $E_1,\ldots,E_d$.
	\end{lemma}
	
	\begin{proof}
		Set $E_0=1$.  Newton's identities give the recursion
		\begin{equation}
		jE_j
		=
		\sum_{r=1}^{j}(-1)^{r-1}E_{j-r}m_r,
		\qquad 1\le j\le d.
		\label{eq:newton-recursion}
		\end{equation}
		Because $1,\ldots,d$ are invertible in~$\F$, this determines the $E_j$ uniquely and efficiently.  The identities hold for every finite multiset, hence for every subset.
	\end{proof}
	
	\begin{proposition}[MSS to Reed--Solomon BDD]
		\label{prop:mss-to-rs}
		Let $A=\{a_1,\ldots,a_{n-1}\}\subseteq\F^\times$ and suppose
		$\operatorname{char}(\F)>d$.  An $\MSS(d)$ instance on~$A$ with required subset size~$k\geq d$ reduces in polynomial time to an \RSBDD{}$(d)$ instance of block length~$n$ and dimension
		\begin{equation}
		K=k-d+1.
		\label{eq:rs-dimension-from-mss}
		\end{equation}
	\end{proposition}
	
	\begin{proof}
		Use \Cref{lem:newton-conversion} to obtain elementary-symmetric targets
		$E_1,\ldots,E_d$.  Define
		\begin{equation}
		f(T)
		=
		T^d-E_1T^{d-1}+E_2T^{d-2}
		-\cdots+(-1)^{d-1}E_{d-1}T.
		\label{eq:f-polynomial}
		\end{equation}
		The evaluation set is
		\begin{equation}
		\D=\{a_1^{-1},\ldots,a_{n-1}^{-1},0\}.
		\end{equation}
		At the nonzero position $a_i^{-1}$ put the received value $-f(a_i)$, and at~$0$ put $(-1)^dE_d$.
		
		Suppose first that $S\subseteq A$ is a solution of size~$k$, and let
		\begin{equation}
		g(T)=\prod_{a\in S}(T-a).
		\end{equation}
		The leading $d$ coefficients of~$g$ are determined by
		$E_1,\ldots,E_d$.  Consequently,
		\begin{equation}
		p(Z)=Z^{k-d}g(Z^{-1})-f(Z^{-1})
		\label{eq:p-from-g}
		\end{equation}
		contains no negative powers and has degree at most $k-d=K-1$.  Its constant coefficient is $(-1)^dE_d$.  For every $a\in S$,
		\begin{equation}
		p(a^{-1})=-f(a),
		\end{equation}
		so $p$ agrees with the received word in the $k$ positions indexed by~$S$ and also at~$0$.  Thus it has $k+1=K+d$ agreements.
		
		Conversely, suppose a polynomial $p$ of degree at most $k-d$ has at least $k+1$ agreements.  It must agree at~$0$.  Indeed, if all $k+1$ agreements were at nonzero points, then
		\begin{equation}
		G(T)=T^{k-d}\bigl(p(T^{-1})+f(T)\bigr)
		\label{eq:G-converse}
		\end{equation}
		would be a monic polynomial of degree~$k$ with $k+1$ distinct roots, a contradiction.  Hence $p(0)=(-1)^dE_d$, and at least $k$ nonzero positions agree.  For their corresponding set~$S$, the polynomial~$G$ in \eqref{eq:G-converse} has precisely the prescribed first $d$ leading coefficients.  Its roots in~$S$ therefore have elementary symmetric functions $E_1,\ldots,E_d$.  By \Cref{lem:newton-conversion}, $S$ solves the original moment instance.
	\end{proof}
	
	For $d=1$, the \RSBDD{} radius in \Cref{prop:mss-to-rs} is one below the covering radius.  This makes explicit the passage from finite-field subset sum to non-deep-hole recognition.
	
	\subsection{Exact 1-in-3-SAT and its vector subset-sum encoding}
	\label{subsec:vector-subset-sum}
	
	We use \emph{cubic planar monotone Exact 1-in-3-SAT}, which is NP-complete by Moore and Robson~\cite{MooreRobson2001}. Each clause contains three distinct Boolean variables, with no negations, and every variable occurs in exactly three clauses. Counting variable occurrences in two ways shows that the number of clauses equals the number of variables; let $m$ denote this common number. The variable--clause incidence graph is planar; our reduction does not use planarity. Define the exact-one relation
	\begin{equation}
	\operatorname{ONE}_3(x,y,z)=1
	\quad\Longleftrightarrow\quad
	x+y+z=1
	\qquad (x,y,z\in\{0,1\}).
	\label{eq:one3-relation}
	\end{equation}
	An instance is a conjunction
	\begin{equation}
	\varphi
	=
	\bigwedge_{r=1}^{m}
	\operatorname{ONE}_3
	\bigl(z_{i(r,1)},z_{i(r,2)},z_{i(r,3)}\bigr)
	\label{eq:exact-one-formula}
	\end{equation}
	on variables $z_1,\ldots,z_m$. The question is whether there is an assignment $b=(b_1,\ldots,b_m)\in\{0,1\}^m$ for which exactly one variable is true in each clause. We may restrict to nonempty instances, since the empty instance is trivially satisfiable.
	
	The first $m$ coordinates of the vector encoding record the choice made for each variable, and the last $m$ coordinates count true variables in the clauses. Put
	\begin{equation}
	t=2m.
	\label{eq:vector-dimension}
	\end{equation}
	For a clause $C_r$ and $(i,b)\in[m]\times\{0,1\}$, let
	\begin{equation}
	\chi_r(i,b)=b\,\mathbf{1}[z_i\in C_r].
	\label{eq:variable-satisfaction-indicator}
	\end{equation}
	Thus $\chi_r(i,0)=0$, and $\chi_r(i,1)=1$ exactly when $z_i$ occurs in~$C_r$.
	
	Let $e_i\in\F_q^m$ denote the $i$th standard basis vector.  Associate with the Boolean choice $z_i=b$ the vector
	\begin{equation}
	v_{i,b}
	=
	\bigl(
	e_i\,;
	\chi_1(i,b),\ldots,\chi_m(i,b)
	\bigr)
	\in\F_q^{2m}.
	\label{eq:boolean-choice-vector}
	\end{equation}
	Equivalently, if $a_i\in\{0,1\}^m$ is the clause-incidence vector of~$z_i$, then $v_{i,0}=(e_i;0)$ and $v_{i,1}=(e_i;a_i)$. Both Boolean choices are retained even though the clauses contain no negations. The vectors $v_{i,b}$ are pairwise distinct.  If $i\ne i'$, then their variable-coordinate blocks are the distinct standard basis vectors $e_i$ and $e_{i'}$.  If $i=i'$ and $b\ne b'$, choose a clause containing~$z_i$; exactly one of the two values $b=0,1$ makes that occurrence true, so the corresponding clause coordinate distinguishes $v_{i,0}$ from $v_{i,1}$.  Each $v_{i,b}$ is also nonzero because its variable-coordinate block is~$e_i$.
	
	The target vector is
	\begin{equation}
	T
	=
	\bigl(
	\one_m\,;
	\one_m
	\bigr)
	\in\F_q^{2m}.
	\label{eq:boolean-target-vector}
	\end{equation}
	
	\begin{proposition}[Exact 1-in-3-SAT as vector subset sum]
		\label{prop:one-in-three-vector-sum}
		Assume that $q$ is prime and $q\ge3$.  For a set
		$I\subseteq[m]\times\{0,1\}$, one has
		\begin{equation}
		\sum_{(i,b)\in I}v_{i,b}=T
		\label{eq:vector-subset-sum}
		\end{equation}
		if and only if $I$ contains exactly one pair $(i,b_i)$ for every variable~$i$, and the resulting assignment $(b_1,\ldots,b_m)$ satisfies~$\varphi$ in the exact 1-in-3 sense.  In particular, every solution of \eqref{eq:vector-subset-sum} has cardinality~$m$.
	\end{proposition}
	
	\begin{proof}
		Consider first the coordinate associated with a variable~$z_i$.  The only vectors that can contribute to this coordinate are $v_{i,0}$ and $v_{i,1}$, and each contributes~$1$.  Hence the coordinate equals
		\begin{equation}
		\bigl|I\cap\{(i,0),(i,1)\}\bigr|
		\in\{0,1,2\}.
		\end{equation}
		Because $\operatorname{char}(\F_q)=q\ge3$, the elements $0,1,2$ are distinct in~$\F_q$.  Equality with the target coordinate~$1$ is therefore equivalent to selecting exactly one of $(i,0)$ and $(i,1)$.  Doing this for every~$i$ defines a Boolean assignment $b_i$ and also forces $|I|=m$.
		
		For a clause~$C_r$, the corresponding coordinate of the sum is
		\begin{equation}
		\sum_{i=1}^{m}\chi_r(i,b_i).
		\label{eq:clause-coordinate-count}
		\end{equation}
		Since $C_r$ contains three distinct variables, this integer is exactly the number of true variables in the clause and belongs to $\{0,1,2,3\}$.  For every prime $q\ge3$, the only integer in $\{0,1,2,3\}$ congruent to~$1$ modulo~$q$ is~$1$. In particular, when $q=3$, a clause with three true variables contributes~$0$, not the target~$1$. Thus the clause coordinate equals the target precisely when exactly one of its three variables is true.  This proves both directions.
	\end{proof}
	
	\paragraph{How the encoding is used later.}
	The final MSS instance lives in an extension field~$\Kfield$ equipped with an $\F_q$-linear quotient map
	\begin{equation}
	\pi:\Kfield\longrightarrow\F_q^t.
	\end{equation}
	For each Boolean choice we take $\ell_{i,b}=\sigma(v_{i,b})$, where $\sigma$ is an injective linear section of~$\pi$.  Hence the Boolean-choice elements are pairwise distinct and satisfy
	$\pi(\ell_{i,b})=v_{i,b}$.  All auxiliary moment-completion elements lie in $\ker\pi$.  Consequently, applying~$\pi$ to the \emph{first} moment equation of any MSS solution recovers exactly \eqref{eq:vector-subset-sum}.  The proposition then supplies the satisfying assignment.  The higher moment equations are used only for completeness: the auxiliary pool fills the residual moments left by a satisfying assignment.
	
	\section{Earlier Reed--Solomon BDD hardness and the PTE bottleneck}
	\label{sec:history}
	
	Gandikota, Ghazi, and Grigorescu proved the first deterministic Reed--Solomon BDD NP-hardness result with a growing gap below the covering radius by reducing through moments subset sum~\cite{GandikotaGhaziGrigorescu2018}.  Their deterministic construction produces an MSS instance of size roughly
	\begin{equation}
	|A|=m\bigl(2^{d+1}-2\bigr),
	\label{eq:ggg-size}
	\end{equation}
	where~$m$ denotes the number of Boolean variables in that construction.  The exponential factor comes from assignment-specific inhomogeneous PTE gadgets.  For two literal values $a$ and~$b$, their auxiliary sets $X$ and~$Y$ must satisfy
	\begin{align}
	\sum_{x\in X}x
	&=\sum_{y\in Y}y, \\
	a^j+\sum_{x\in X}x^j
	&=b^j+\sum_{y\in Y}y^j,
	&&2\le j\le d.
	\label{eq:ggg-inhomogeneous-pte}
	\end{align}
	Their general explicit construction has size $2^{\Theta(d)}$.
	
	There is a second difficulty.  A purported MSS witness may select arbitrary pieces of many auxiliary gadgets rather than one complete bundle per variable.  The earlier proof controls this by arranging enormous separations of magnitude: every auxiliary sub-subset has either a tiny or a huge first-moment sum.  This is the bimodality condition.
	
	The same paper proves, nonconstructively, that signed inhomogeneous PTE systems over sufficiently large finite fields have solutions of size $O(d)$ when
	$d<|\F|^{1/2-\delta}$.  That existence theorem does not by itself improve the reduction: the gadgets were still required to be explicit, assignment-labeled, and sound against unintended sub-subsets.
	
	Our construction separates these roles.
	\begin{itemize}
		\item The arithmetic input is only \emph{existence} of a one-sided completion for every target.
		\item The reduction constructs one universal pool rather than one bundle for every Boolean choice.
		\item Soundness is enforced by a quotient map, not by arithmetic restrictions on auxiliary sub-subsets.
	\end{itemize}
	This is where the point-count estimates in~\cite{Wan2026} become useful in a complexity reduction even though they do not directly provide an algorithm for finding the completing subset.
	
	\section{Uniform positive moment completion}
	\label{sec:completion}
	
	This section gives two versions of the arithmetic completion theorem.  The sharper version is a direct consequence of the higher-dimensional point-count argument in~\cite{Wan2026}; we state it without reproducing the Deligne--Hooley--Katz machinery.  We then prove a slightly weaker version using only the classical one-variable Weil bound for additive character sums~\cite{Schmidt1976}.  The weaker version is the sole completion theorem used in the NP-hardness reduction.
	
	\subsection{Notation and the higher-dimensional estimate}
	
	Let $q$ be prime, let $1\le d<h<q$, and fix
	$\boldsymbol{c}=(c_1,\ldots,c_d)\in\F_q^d$.  Define
	\begin{equation}
	X_{\boldsymbol{c},h}
	=
	\left\{
	(x_1,\ldots,x_h)\in\F_q^h:
	\sum_{i=1}^{h}x_i^j=c_j
	\text{ for }1\le j\le d
	\right\}.
	\label{eq:X-c-h}
	\end{equation}
	For $r\ne s$, put
	\begin{equation}
	Y_{r,s}
	=X_{\boldsymbol{c},h}\cap\{x_r=x_s\},
	\label{eq:Y-r-s}
	\end{equation}
	and for $1\le r\le h$ put
	\begin{equation}
	Z_r
	=X_{\boldsymbol{c},h}\cap\{x_r=0\}.
	\label{eq:Z-r}
	\end{equation}
	
	\begin{theorem}[Uniform higher-dimensional point-count estimates]
		\label{thm:uniform-higher-dimensional}
		Suppose $1\le d\le h-3$ and $h<q$.  Uniformly for every
		$\boldsymbol{c}\in\F_q^d$,
		\begin{equation}
		\left|
		|X_{\boldsymbol{c},h}|-q^{h-d}
		\right|
		\le
		\frac12(2d+2)^h q^{(h-d+1)/2},
		\label{eq:uniform-total-count}
		\end{equation}
		and, for every $r\ne s$,
		\begin{equation}
		\left|
		|Y_{r,s}|-q^{h-d-1}
		\right|
		\le
		\frac12(2d+2)^{h-1}q^{(h-d+1)/2}.
		\label{eq:uniform-collision-count}
		\end{equation}
	\end{theorem}
	
	With $k=d+1$, these are the estimates proved in
	\cite[Propositions~5.2 and~6.2]{Wan2026}.  The same proof applies to an arbitrary right-hand side $\boldsymbol{c}$: the target values occur only in lower-order affine terms, while the projective section at infinity, the degree data, and the Betti-number bounds are unchanged.  Thus the constants in that proof are target-uniform.  We do not repeat the higher-dimensional cohomological argument here; it uses the Deligne--Hooley--Katz estimate and the Betti-number bound in~\cite{WanZhang2026}.
	
	\subsection{An elementary one-variable point-count estimate}
	
	The next result is weaker than \Cref{thm:uniform-higher-dimensional}, but its proof uses only one-variable character sums.
	
	\begin{theorem}[Elementary point-count bounds]
		\label{thm:elementary-point-count}
		Let $q$ be prime, let $2\le d<q$, and let $h\ge2d+2$.  Then, uniformly for every $\boldsymbol{c}\in\F_q^d$,
		\begin{align}
		\left|
		|X_{\boldsymbol{c},h}|-q^{h-d}
		\right|
		&\le
		d!(d-1)^{h-2d}q^{h/2},
		\label{eq:elementary-total-count}\\
		\left|
		|Z_r|-q^{h-d-1}
		\right|
		&\le
		d!(d-1)^{h-2d-1}q^{(h-1)/2},
		\label{eq:elementary-zero-count}\\
		\left|
		|Y_{r,s}|-q^{h-d-1}
		\right|
		&\le
		d!(d-1)^{h-2d-1}q^{(h-1)/2}
		\label{eq:elementary-collision-count}
		\end{align}
		for every $r$ and every $r\ne s$.
	\end{theorem}
	
	\begin{proof}
		Fix a nontrivial additive character
		$\psi:\F_q\to\mathbb C^\times$.  For
		$\boldsymbol{\lambda}=(\lambda_1,\ldots,\lambda_d)\in\F_q^d$, define
		\begin{equation}
		F_{\boldsymbol{\lambda}}(T)
		=\sum_{j=1}^{d}\lambda_jT^j,
		\qquad
		W(\boldsymbol{\lambda})
		=\sum_{x\in\F_q}\psi\!\left(F_{\boldsymbol{\lambda}}(x)\right).
		\label{eq:weil-W-definition}
		\end{equation}
		If $\boldsymbol{\lambda}\ne0$, then
		$F_{\boldsymbol{\lambda}}$ is nonconstant of degree at most~$d<q$.  The classical Weil bound gives
		\begin{equation}
		|W(\boldsymbol{\lambda})|
		\le(d-1)\sqrt q.
		\label{eq:one-variable-weil}
		\end{equation}
		For a linear polynomial the sum is zero, so the same bound remains valid.
		
		Additive-character orthogonality gives
		\begin{equation}
		|X_{\boldsymbol{c},h}|
		=q^{-d}
		\sum_{\boldsymbol{\lambda}\in\F_q^d}
		\psi(-\boldsymbol{\lambda}\mathbin{\cdot}\boldsymbol{c})
		W(\boldsymbol{\lambda})^h.
		\label{eq:fourier-fiber-count}
		\end{equation}
		The zero frequency contributes $q^{h-d}$.  To control all nonzero frequencies together, expand the moment of order~$2d$:
		\begin{equation}
		\sum_{\boldsymbol{\lambda}\in\F_q^d}
		|W(\boldsymbol{\lambda})|^{2d}
		=q^d T_d,
		\label{eq:2d-moment-identity}
		\end{equation}
		where $T_d$ is the number of pairs of ordered $d$-tuples
		$\boldsymbol{x},\boldsymbol{y}\in\F_q^d$ satisfying
		\begin{equation}
		\sum_{i=1}^{d}x_i^j
		=
		\sum_{i=1}^{d}y_i^j,
		\qquad 1\le j\le d.
		\label{eq:equal-d-moments}
		\end{equation}
		Indeed, \eqref{eq:2d-moment-identity} follows by expanding
		$W(\boldsymbol{\lambda})^d\overline{W(\boldsymbol{\lambda})}^d$ and summing first over $\boldsymbol{\lambda}$.
		
		Because $q>d$, Newton identities recover the elementary symmetric functions of a $d$-tuple from its first $d$ power sums.  Hence the two tuples in \eqref{eq:equal-d-moments} have the same multiset of coordinates.  For each fixed $\boldsymbol{x}$, there are at most $d!$ possible ordered tuples $\boldsymbol{y}$.  Therefore
		\begin{equation}
		T_d\le d!q^d,
		\qquad
		\sum_{\boldsymbol{\lambda}}
		|W(\boldsymbol{\lambda})|^{2d}
		\le d!q^{2d}.
		\label{eq:2d-moment-bound}
		\end{equation}
		Using \eqref{eq:one-variable-weil} for the remaining $h-2d$ powers gives
		\begin{align}
		\sum_{\boldsymbol{\lambda}\ne0}
		|W(\boldsymbol{\lambda})|^h
		&\le
		\bigl((d-1)\sqrt q\bigr)^{h-2d}
		\sum_{\boldsymbol{\lambda}}
		|W(\boldsymbol{\lambda})|^{2d} \\
		&\le
		d!(d-1)^{h-2d}q^{h/2+d}.
		\label{eq:high-moment-W-bound}
		\end{align}
		Combining this with \eqref{eq:fourier-fiber-count} proves
		\eqref{eq:elementary-total-count}.
		
		Deleting a coordinate constrained to be zero identifies $Z_r$ with the same moment fiber in $h-1$ variables.  Applying the preceding argument with $h-1$ proves
		\eqref{eq:elementary-zero-count}.
		
		For a collision, set
		\begin{equation}
		W_2(\boldsymbol{\lambda})
		=\sum_{x\in\F_q}
		\psi\!\left(2F_{\boldsymbol{\lambda}}(x)\right).
		\end{equation}
		Since $q$ is odd, the polynomial $2F_{\boldsymbol{\lambda}}$ has the same degree as $F_{\boldsymbol{\lambda}}$, and
		\begin{equation}
		|W_2(\boldsymbol{\lambda})|
		\le(d-1)\sqrt q
		\qquad
		(\boldsymbol{\lambda}\ne0).
		\label{eq:W2-weil}
		\end{equation}
		After imposing $x_r=x_s$, character orthogonality gives
		\begin{equation}
		|Y_{r,s}|
		=q^{-d}
		\sum_{\boldsymbol{\lambda}\in\F_q^d}
		\psi(-\boldsymbol{\lambda}\mathbin{\cdot}\boldsymbol{c})
		W_2(\boldsymbol{\lambda})
		W(\boldsymbol{\lambda})^{h-2}.
		\label{eq:fourier-collision-count}
		\end{equation}
		The zero frequency contributes $q^{h-d-1}$.  For the remaining frequencies, use \eqref{eq:W2-weil}, the pointwise bound for $h-2-2d$ further powers of $W$, and \eqref{eq:2d-moment-bound}.  This gives
		\begin{align}
		&q^{-d}
		\sum_{\boldsymbol{\lambda}\ne0}
		|W_2(\boldsymbol{\lambda})|
		|W(\boldsymbol{\lambda})|^{h-2} \\
		&\qquad\le
		d!(d-1)^{h-2d-1}q^{(h-1)/2},
		\end{align}
		which proves \eqref{eq:elementary-collision-count}.
	\end{proof}
	
	\subsection{Strong and elementary completion theorems}
	
	\begin{theorem}[Strong uniform positive moment completion]
		\label{thm:strong-uniform-completion}
		Fix $\rho>0$ and choose a constant
		\begin{equation}
		C>1+\frac{2}{\rho}.
		\label{eq:strong-C-condition}
		\end{equation}
		For all sufficiently large~$d$, let $q$ be a prime satisfying
		\begin{equation}
		q\ge d^{2+\rho},
		\end{equation}
		and put $h=\lceil Cd\rceil$.  Then every
		$\boldsymbol{c}=(c_1,\ldots,c_d)\in\F_q^d$ has an $h$-element subset
		$S\subseteq\F_q^\times$ such that
		\begin{equation}
		\sum_{a\in S}a^j=c_j,
		\qquad 1\le j\le d.
		\label{eq:strong-completion-conclusion}
		\end{equation}
	\end{theorem}
	
	This sharper theorem follows directly by combining the target-uniform estimates in \Cref{thm:uniform-higher-dimensional} with the distinct-coordinate sieve in~\cite{Wan2026}; equivalently, one repeats the union bound in the proof of \cite[Theorem~6.3]{Wan2026}, also applying the same point-count estimate to the zero-coordinate sections.  We omit the repetition because the NP-hardness proof below uses only the elementary version that follows.
	
	\begin{theorem}[Elementary uniform positive moment completion]
		\label{thm:uniform-completion}
		Fix $\rho>0$ and choose a constant
		\begin{equation}
		C>2+\frac{2}{\rho}.
		\label{eq:C-condition}
		\end{equation}
		For all sufficiently large~$d$, let $q$ be a prime satisfying
		\begin{equation}
		q\ge d^{2+\rho},
		\label{eq:q-completion-size}
		\end{equation}
		and put $h=\lceil Cd\rceil$.  Then every
		$\boldsymbol{c}=(c_1,\ldots,c_d)\in\F_q^d$ has an $h$-element subset
		$S\subseteq\F_q^\times$ such that
		\begin{equation}
		\sum_{a\in S}a^j=c_j,
		\qquad 1\le j\le d.
		\label{eq:completion-conclusion}
		\end{equation}
	\end{theorem}
	
	\begin{proof}
		For sufficiently large~$d$, the assumptions imply
		$h\ge2d+2$ and $h<q$.  Let $X=X_{\boldsymbol{c},h}$, and let
		$X^\circ$ be the set of ordered solutions with nonzero, pairwise-distinct coordinates.  By the union bound and \Cref{thm:elementary-point-count},
		\begin{align}
		|X^\circ|
		\ge{}& q^{h-d}
		-d!(d-1)^{h-2d}q^{h/2} \\
		&-\left(h+\binom h2\right)
		\left(
		q^{h-d-1}
		+d!(d-1)^{h-2d-1}q^{(h-1)/2}
		\right).
		\label{eq:elementary-union-bound}
		\end{align}
		After division by $q^{h-d}$, set
		\begin{align}
		E_0
		&=d!(d-1)^{h-2d}q^{d-h/2},
		\label{eq:E0-definition}\\
		E_1
		&=d!(d-1)^{h-2d-1}q^{d-(h+1)/2}.
		\label{eq:E1-definition}
		\end{align}
		Then
		\begin{equation}
		\frac{|X^\circ|}{q^{h-d}}
		\ge
		1-E_0
		-\left(h+\binom h2\right)
		\left(q^{-1}+E_1\right).
		\label{eq:elementary-normalized-count}
		\end{equation}
		
		Using $d!\le d^d$ and $(d-1)^{h-2d}\le d^{h-2d}$, we obtain
		\begin{align}
		E_0
		&\le d^{h-d}q^{d-h/2} \\
		&\le d^{(1+\rho)d-\rho h/2}.
		\label{eq:E0-decay}
		\end{align}
		The second inequality uses $d-h/2<0$ and
		$q\ge d^{2+\rho}$.  Since $h\ge Cd$ and
		$C>2+2/\rho$, the exponent in \eqref{eq:E0-decay} is at most
		\begin{equation}
		-\left(
		\frac{\rho C}{2}-1-\rho
		\right)d,
		\end{equation}
		which is negative with linear magnitude.  Hence $E_0\to0$ faster than any inverse polynomial in~$d$.
		
		Moreover,
		\begin{equation}
		E_1
		=\frac{E_0}{(d-1)\sqrt q}.
		\label{eq:E1-vs-E0}
		\end{equation}
		Because $h=O_\rho(d)$ and $q\ge d^{2+\rho}$,
		\begin{equation}
		\left(h+\binom h2\right)q^{-1}
		=O_\rho(d^{-\rho}),
		\label{eq:main-sieve-error}
		\end{equation}
		and
		\begin{equation}
		\left(h+\binom h2\right)E_1
		=E_0\,O_\rho(d^{-\rho/2}).
		\label{eq:character-sieve-error}
		\end{equation}
		All error terms on the right-hand side of
		\eqref{eq:elementary-normalized-count} therefore tend to zero.  For sufficiently large~$d$, the right-hand side is positive, so $X^\circ$ contains an ordered tuple of distinct nonzero coordinates.  Forgetting the order gives the required subset~$S$.
	\end{proof}
	
	\begin{remark}[PTE interpretation]
		If two distinct $h$-subsets realize the same target, canceling their intersection produces disjoint sets with equal first $d$ moments.  Thus a moment fiber containing two distinct subsets yields a finite-field PTE configuration.  The reduction below, however, uses the stronger prescribed-target completion property rather than a homogeneous PTE pair.
	\end{remark}
	
	\section{The extension-field quotient construction and power-spanning directions}
	\label{sec:extension-algebra}
	
	Let $t$ be the dimension of the vector subset-sum instance from
	\Cref{subsec:vector-subset-sum}.  We now build the extension field in which the MSS instance will live.
	
	\subsection{Explicit field construction and a product-spanning subspace}
	\label{subsec:explicit-extension-field}
	
	Let
	\begin{equation}
	L=2t+1=4m+1.
	\label{eq:extension-degree}
	\end{equation}
	We first make the construction and representation of the extension field explicit.  In the final reduction, $q$ is a prime satisfying
	\begin{equation}
	q=\Theta\!\left(d^{2+\rho}\right),
	\qquad
	d=t^{\nu},
	\label{eq:q-polynomial-preview}
	\end{equation}
	for a fixed positive integer $\nu$ and a fixed $\rho>0$; see \Cref{sec:hardness}.  Hence
	$q=t^{O_{\nu,\rho}(1)}$.
	
	By Shoup's deterministic algorithm, one can construct a monic irreducible polynomial
	\begin{equation}
	F(Z)\in\F_q[Z],
	\qquad
	\deg F=L,
	\label{eq:irreducible-polynomial}
	\end{equation}
	in time
	\begin{equation}
	O(L^4q)
	\label{eq:shoup-running-time}
	\end{equation}
	for the parameters used here~\cite{Shoup1990}. This coarse bound suffices: since $L=2t+1$ and $q=t^{O_{\nu,\rho}(1)}$, the extension field is constructed deterministically in polynomial time in the source-instance size.
	
	Set
	\begin{equation}
	\Kfield
	=\F_q[Z]/(F(Z)),
	\qquad
	\theta=Z\bmod F(Z).
	\label{eq:extension-field}
	\end{equation}
	Then
	\begin{equation}
	1,\theta,\ldots,\theta^{2t}
	\label{eq:power-basis}
	\end{equation}
	is an $\F_q$-basis of~$\Kfield$.  We represent a field element by its coefficient vector in this basis.  Addition is coefficientwise, multiplication is polynomial multiplication followed by reduction modulo~$F$, and inversion is computed by the extended Euclidean algorithm in~$\F_q[Z]$.  These operations take time polynomial in $L\log q$, and each element has representation length
	\begin{equation}
	O(L\log q)=O(t\log d)
	\label{eq:field-element-size}
	\end{equation}
	bits for the parameters of the reduction.
	
	Define
	\begin{equation}
	H=\Span_{\F_q}\{1,\theta,\ldots,\theta^t\}.
	\label{eq:H-definition}
	\end{equation}
	Then
	\begin{equation}
	\dim_{\F_q}(\Kfield/H)=t,
	\label{eq:quotient-dimension}
	\end{equation}
	and
	\begin{equation}
	\Span_{\F_q}(H\cdot H)=\Kfield,
	\label{eq:H-product-spans}
	\end{equation}
	because the products of basis monomials span
	$1,\theta,\ldots,\theta^{2t}$.
	
	Write each field element uniquely as
	\begin{equation}
	x=\sum_{r=0}^{2t}x_r\theta^r,
	\qquad x_r\in\F_q,
	\end{equation}
	and define the quotient map explicitly by
	\begin{equation}
	\pi:\Kfield\longrightarrow\F_q^t,
	\qquad
	\pi(x)=(x_{t+1},\ldots,x_{2t}).
	\label{eq:quotient-isomorphism}
	\end{equation}
	Then $\ker\pi=H$.  The linear section
	\begin{equation}
	\sigma(u_1,\ldots,u_t)
	=\sum_{r=1}^{t}u_r\theta^{t+r}
	\end{equation}
	is explicit and injective because $\pi\circ\sigma$ is the identity on~$\F_q^t$.  For each Boolean choice $(i,b)$, define
	\begin{equation}
	\ell_{i,b}=\sigma(v_{i,b}).
	\label{eq:choice-lift}
	\end{equation}
	The vectors $v_{i,b}$ are pairwise distinct, so the Boolean-choice elements $\ell_{i,b}$ are pairwise distinct as well.  Moreover,
	\begin{equation}
	\pi(\ell_{i,b})=v_{i,b}.
	\end{equation}
	Since every $v_{i,b}$ is nonzero, all Boolean-choice elements lie outside~$H=\ker\pi$.
	
	Every auxiliary element will lie in~$H$.  Consequently, for every auxiliary subset $Y\subseteq H$,
	\begin{equation}
	\pi\!\left(\sum_{y\in Y}y\right)=0.
	\label{eq:quotient-zero}
	\end{equation}
	This quotient identity protects soundness.  Notice that it holds for \emph{every} auxiliary sub-subset; no bundle structure or bimodality is needed.
	
	\subsection{A deterministic power condenser of linear size}
	\label{subsec:deterministic-power-condenser}
	
	We next construct a projectively distinct set of directions whose powers span
	simultaneously.  The construction has linear size and is fully deterministic.
	The condition on~$q$ below is automatic in the reduction because
	$q\ge d^{2+\rho}$.
	
	\begin{lemma}[Pure powers span the extension]
		\label{lem:pure-powers-span}
		Assume $2\le j<q$.  Then
		\begin{equation}
		\Span_{\F_q}\{x^j:x\in H\}=\Kfield.
		\label{eq:pure-powers-span}
		\end{equation}
	\end{lemma}
	
	\begin{proof}
		Let
		\begin{equation}
		H^{\langle j\rangle}
		=
		\Span_{\F_q}
		\{x_1\cdots x_j:x_1,\ldots,x_j\in H\}.
		\end{equation}
		Because $1\in H$ and \eqref{eq:H-product-spans} holds,
		\begin{equation}
		\Kfield
		=\Span(H\cdot H)
		\subseteq H^{\langle j\rangle}
		\subseteq\Kfield,
		\end{equation}
		so $H^{\langle j\rangle}=\Kfield$.  Since $j!\ne0$ in~$\F_q$, polarization gives
		\begin{equation}
		j!\,x_1\cdots x_j
		=
		\sum_{S\subseteq[j]}
		(-1)^{j-|S|}
		\left(\sum_{i\in S}x_i\right)^j.
		\label{eq:full-polarization}
		\end{equation}
		Every sum on the right belongs to~$H$.  Hence every $j$-fold product lies in the span of the pure $j$th powers, proving the claim.
	\end{proof}
	
	For an $\F_q$-subspace $U\subseteq\Kfield$, define
	\begin{equation}
	U^{\langle0\rangle}=\F_q\cdot1,
	\qquad
	U^{\langle s\rangle}
	=
	\Span_{\F_q}
	\{u_1\cdots u_s:u_i\in U\}
	\quad(s\ge1).
	\label{eq:product-power-space}
	\end{equation}
	These spaces are computed recursively by
	\begin{equation}
	U^{\langle s+1\rangle}
	=
	\Span_{\F_q}
	\{uw:u\in U,\ w\in U^{\langle s\rangle}\}.
	\label{eq:product-power-recursion}
	\end{equation}
	All of them lie in the $L$-dimensional vector space~$\Kfield$.
	
	\begin{lemma}[Identity test on an affine subspace]
		\label{lem:affine-power-identity-test}
		Let $a\in H$, let $U\subseteq H$ be an $\F_q$-subspace, let
		$1\le j<q$, and let $\lambda:\Kfield\to\F_q$ be linear.  Then the function
		\begin{equation}
		y\longmapsto\lambda\bigl((a+y)^j\bigr),
		\qquad y\in U,
		\label{eq:affine-power-polynomial}
		\end{equation}
		vanishes identically on~$U$ if and only if
		\begin{equation}
		\lambda\!\left(a^{j-s}U^{\langle s\rangle}\right)=0
		\qquad
		\text{for every }0\le s\le j.
		\label{eq:affine-identity-criterion}
		\end{equation}
		The condition in \eqref{eq:affine-identity-criterion} can be tested in deterministic time polynomial in $L$, $j$, and $\log q$.
	\end{lemma}
	
	\begin{proof}
		Choose a basis of~$U$ and regard \eqref{eq:affine-power-polynomial} as a polynomial in the corresponding coordinates.  Its total degree is at most~$j<q$.  Expanding by homogeneous degree gives
		\begin{equation}
		\lambda\bigl((a+y)^j\bigr)
		=
		\sum_{s=0}^{j}
		\binom{j}{s}
		\lambda\bigl(a^{j-s}y^s\bigr).
		\label{eq:affine-binomial-expansion}
		\end{equation}
		A polynomial of degree less than~$q$ in every variable that vanishes on the full grid $\F_q^{\dim U}$ is the zero polynomial.  Hence the polynomial in \eqref{eq:affine-binomial-expansion} vanishes identically exactly when each homogeneous component vanishes.  The binomial coefficients are nonzero, and polarization implies
		\begin{equation}
		\Span_{\F_q}\{y^s:y\in U\}=U^{\langle s\rangle}.
		\end{equation}
		This proves the equivalence.  The recursive construction in
		\eqref{eq:product-power-recursion}, followed by multiplication by
		$a^{j-s}$ and Gaussian elimination, gives the claimed deterministic test.
	\end{proof}
	
	\begin{lemma}[Deterministic linear-size power condenser]
		\label{lem:deterministic-directions}
		Assume $d\ge2$ and
		\begin{equation}
		q>\sum_{j=1}^{d}j=\frac{d(d+1)}2.
		\label{eq:condenser-field-size}
		\end{equation}
		One can construct, in deterministic time polynomial in $L$, $d$, and
		$\log q$, a projectively distinct set
		$B\subseteq H\setminus\{0\}$ satisfying
		\begin{align}
		\Span_{\F_q}B&=H,
		\label{eq:B-first-span}\\
		\Span_{\F_q}\{\beta^j:\beta\in B\}&=\Kfield,
		&&2\le j\le d,
		\label{eq:B-higher-span}
		\end{align}
		and
		\begin{equation}
		|B|=L+1.
		\label{eq:B-linear-size}
		\end{equation}
		The set~$B$ has even cardinality.
	\end{lemma}
	
	\begin{proof}
		We first construct a projectively distinct set $B_0$ of exactly~$L$ directions.  Suppose that $B_r$ has already been selected, and define
		\begin{align}
		W_{1,r}
		&=\Span_{\F_q}B_r\subseteq H,\\
		W_{j,r}
		&=\Span_{\F_q}\{\beta^j:\beta\in B_r\}
		\subseteq\Kfield,
		&&2\le j\le d.
		\label{eq:current-power-spans}
		\end{align}
		For each deficient span choose an annihilating functional.  If
		$W_{1,r}\ne H$, choose a nonzero functional on $H/W_{1,r}$ and extend it to a linear functional
		$\lambda_1:\Kfield\to\F_q$.  Thus
		\begin{equation}
		\lambda_1(W_{1,r})=0,
		\qquad
		\lambda_1|_H\ne0.
		\label{eq:first-annihilator}
		\end{equation}
		If $W_{j,r}\ne\Kfield$ for $j\ge2$, choose
		$0\ne\lambda_j\in\operatorname{Hom}_{\F_q}(\Kfield,\F_q)$ with
		\begin{equation}
		\lambda_j(W_{j,r})=0.
		\label{eq:higher-annihilator}
		\end{equation}
		Call these indices active.  By \eqref{eq:first-annihilator} and
		Lemma~\ref{lem:pure-powers-span}, for every active~$j$ the polynomial
		\begin{equation}
		x\longmapsto\lambda_j(x^j),
		\qquad x\in H,
		\label{eq:active-power-polynomial}
		\end{equation}
		is not identically zero.
		
		We now find one $\beta\in H$ on which all active polynomials are nonzero.  Let
		$e_1,\ldots,e_{t+1}$ be a basis of~$H$.  Fix the coordinates of
		\begin{equation}
		\beta=c_1e_1+\cdots+c_{t+1}e_{t+1}
		\label{eq:coordinate-fixed-beta}
		\end{equation}
		one at a time.  Suppose $c_1,\ldots,c_{s-1}$ have been fixed, put
		\begin{equation}
		a=\sum_{i<s}c_ie_i,
		\qquad
		U=\Span\{e_s,\ldots,e_{t+1}\},
		\end{equation}
		and maintain that every active polynomial remains nonzero on the affine space $a+U$.  Set
		\begin{equation}
		U'=\Span\{e_{s+1},\ldots,e_{t+1}\}.
		\end{equation}
		For a formal scalar~$C$, consider
		\begin{equation}
		G_j(C,y)
		=
		\lambda_j\bigl((a+Ce_s+y)^j\bigr),
		\qquad y\in U'.
		\label{eq:coordinate-restriction}
		\end{equation}
		The induction invariant says that $G_j$ is not the zero polynomial.  Viewed as a polynomial in~$y$, at least one coefficient is therefore a nonzero univariate polynomial in~$C$ of degree at most~$j$.  If the restriction
		\begin{equation}
		y\longmapsto G_j(c,y)
		\end{equation}
		vanishes identically on~$U'$, then the grid argument in
		Lemma~\ref{lem:affine-power-identity-test} forces every coefficient in~$y$ to vanish at~$C=c$.  Hence there are at most~$j$ such values of~$c$.  Across all active indices, the total number of bad values is at most
		\begin{equation}
		\sum_{j=1}^{d}j<q.
		\label{eq:number-bad-scalars}
		\end{equation}
		Test any fixed collection of $1+\sum_{j=1}^d j$ distinct field elements.  By Lemma~\ref{lem:affine-power-identity-test}, each candidate is tested deterministically, and at least one candidate preserves every active polynomial.  Continuing through all coordinates produces a vector $\beta\in H$ satisfying
		\begin{equation}
		\lambda_j(\beta^j)\ne0
		\qquad
		\text{for every active }j.
		\label{eq:simultaneous-rank-increase}
		\end{equation}
		
		Set $B_{r+1}=B_r\cup\{\beta\}$.  Equation
		\eqref{eq:simultaneous-rank-increase} increases every deficient dimension by one.  The spaces $W_{j,r}$ for $j\ge2$ have ambient dimension~$L$, so after $L$ rounds they all equal~$\Kfield$; the first-power span reaches~$H$ no later.  Since $j=2$ remains deficient before the $L$th rank increase, the construction selects exactly~$L$ elements.
		
		The selected elements are projectively distinct.  Indeed, if a new element were a nonzero base-field multiple of a previous element, then all of its active powers would already lie in the corresponding spaces $W_{j,r}$, contradicting
		\eqref{eq:simultaneous-rank-increase}.  The selected element is also nonzero, since every active polynomial vanishes at zero.
		
		Finally, $L=2t+1$ is odd.  Let $e_1,\ldots,e_{t+1}$ be a basis of~$H$.  If $t\ge2$, the vectors
		\begin{equation}
		e_i
		\quad\text{and}\quad
		e_i+e_j,
		\qquad 1\le i<j\le t+1,
		\label{eq:parity-direction-list}
		\end{equation}
		represent
		\begin{equation}
		(t+1)+\binom{t+1}{2}
		=\frac{(t+1)(t+2)}2
		>2t+1=L
		\end{equation}
		pairwise distinct projective directions.  If $t=1$, also include
		$e_1+2e_2$, which is available because $q>3$.  Since $B_0$ represents only~$L$ projective classes, this explicit list contains a new direction.  Add it to~$B_0$.  This preserves all spans, gives an even set~$B$, and proves
		\eqref{eq:B-linear-size}.
		
		The algorithm uses at most $L$ outer rounds, at most $t+1\le L$ coordinate-fixing stages per round, and $O(d^2)$ field candidates per stage.  All remaining operations are multiplication and Gaussian elimination in the $L$-dimensional space~$\Kfield$, so the running time is polynomial in $L$, $d$, and~$\log q$.
	\end{proof}

	\section{A balanced universal completion pool}
	\label{sec:universal-pool}
	
	Let $B\subset H\setminus\{0\}$ be a projectively distinct direction set of even size~$R$, satisfying \eqref{eq:B-first-span} and \eqref{eq:B-higher-span}.  Partition it into two equal parts,
	\begin{equation}
	B=B_+\sqcup B_-,
	\qquad |B_+|=|B_-|=R/2,
	\end{equation}
	and define signs
	\begin{equation}
	\eps_\beta=
	\begin{cases}
	+1,&\beta\in B_+,\\
	-1,&\beta\in B_-.
	\end{cases}
	\end{equation}
	For each direction, put the full nonzero base-field line into the auxiliary universe:
	\begin{equation}
	D_\beta=\beta\F_q^\times,
	\qquad
	D=\bigsqcup_{\beta\in B}D_\beta.
	\label{eq:auxiliary-pool}
	\end{equation}
	The union is disjoint because the directions are projectively distinct, and
	\begin{equation}
	|D|=R(q-1).
	\label{eq:D-size}
	\end{equation}
	
	Because $1\le j\le d<q-1$,
	\begin{equation}
	\sum_{a\in\F_q^\times}a^j=0.
	\label{eq:full-line-zero}
	\end{equation}
	This permits a complement trick that keeps every completion at exactly half the size of~$D$.
	
	\begin{theorem}[Balanced universal completion]
		\label{thm:balanced-completion}
		Assume the hypotheses of \Cref{thm:uniform-completion}.  For every residual vector
		\begin{equation}
		r_1\in H,
		\qquad
		r_2,\ldots,r_d\in\Kfield,
		\label{eq:residual-domain}
		\end{equation}
		there is a subset $Y\subseteq D$ such that
		\begin{align}
		|Y|&=\frac{|D|}{2},
		\label{eq:Y-half}\\
		\sum_{y\in Y}y^j&=r_j,
		&&1\le j\le d.
		\label{eq:Y-moments}
		\end{align}
	\end{theorem}
	
	\begin{proof}
		Multiplying spanning vectors by the signs $\eps_\beta$ does not change their spans.  Therefore choose coefficients
		$c_{\beta,j}\in\F_q$ satisfying
		\begin{align}
		r_1
		&=\sum_{\beta\in B}
		\eps_\beta\,\beta c_{\beta,1},
		\label{eq:coefficient-first}\\
		r_j
		&=\sum_{\beta\in B}
		\eps_\beta\,\beta^j c_{\beta,j},
		&&2\le j\le d.
		\label{eq:coefficient-higher}
		\end{align}
		For each~$\beta$, apply \Cref{thm:uniform-completion} to
		$(c_{\beta,1},\ldots,c_{\beta,d})$.  Obtain an $h$-subset
		$S_\beta\subseteq\F_q^\times$ with
		\begin{equation}
		\sum_{a\in S_\beta}a^j=c_{\beta,j},
		\qquad 1\le j\le d.
		\end{equation}
		Define
		\begin{equation}
		Y_\beta=
		\begin{cases}
		\beta S_\beta,&\beta\in B_+,\\
		\beta(\F_q^\times\setminus S_\beta),&\beta\in B_-.
		\end{cases}
		\label{eq:Y-beta}
		\end{equation}
		By \eqref{eq:full-line-zero}, the $j$th moment of $Y_\beta$ is
		$\eps_\beta\beta^jc_{\beta,j}$.  Thus
		$Y=\bigsqcup_\beta Y_\beta$ satisfies \eqref{eq:Y-moments}.
		Its size is
		\begin{align}
		|Y|
		&=\frac{R}{2}h
		+\frac{R}{2}(q-1-h) \\
		&=\frac{R(q-1)}{2}
		=\frac{|D|}{2}.
		\label{eq:Y-size-calculation}
		\end{align}
	\end{proof}
	
	The theorem is the precise universal-completion object needed by the reduction.  The pool~$D$ is constructed explicitly.  The subsets $S_\beta$, and hence~$Y$, are not constructed by the reduction; they are part of the existential NP witness in the completeness proof.
	
	\section{The intermediate moments-subset-sum hardness reduction}
	\label{sec:hardness}
	
	We now give the complete reduction from Exact 1-in-3-SAT.
	
	\subsection{Construction of the MSS instance}
	\label{subsec:mss-construction}
	
	\paragraph{Preprocessing and fixed constants.}
	Fix a rational $0<\alpha<1/2$.  Choose once and for all a positive rational $\rho$ and positive integers $\nu,C$ satisfying
	\begin{equation}
	\alpha<\frac{\nu}{1+\nu(2+\rho)},
	\qquad
	C>2+\frac{2}{\rho}.
	\label{eq:effective-reduction-constants}
	\end{equation}
	Such choices exist by first taking $\rho$ small enough and then $\nu$ large enough; see \Cref{subsec:size-running-time}.
	Starting from a nonempty instance of the cubic planar monotone problem in \Cref{subsec:vector-subset-sum}, take the conjunction of disjoint renamed copies of that instance until the number of variables plus clauses is at least a fixed threshold~$t_0$. The new formula is satisfiable if and only if the original formula is satisfiable. This padding preserves monotonicity, three occurrences per variable, planarity of the incidence graph, and equality of the numbers of variables and clauses.  Choose $t_0$ large enough that the preliminary parameter $d=t^{\nu}$ is at least~$2$ and meets the sufficiently-large-$d$ hypothesis of \Cref{thm:uniform-completion}, and that the bounds
	$\lfloor n^\alpha\rfloor\le d$ and
	$\lfloor n^\alpha\rfloor\le(n-1)/2$ hold for every output of the construction below.  The explicit error bounds in \Cref{thm:uniform-completion} and the size estimates in \Cref{subsec:size-running-time} give an effective such threshold depending only on the fixed constants.  For inputs below the fixed threshold, padding adds only a bounded number of variables and clauses. It therefore takes polynomial time, including on the finitely many small inputs.
	
	Let
	\begin{equation}
	\varphi
	=
	\bigwedge_{r=1}^{m}
	\operatorname{ONE}_3
	\bigl(z_{i(r,1)},z_{i(r,2)},z_{i(r,3)}\bigr)
	\label{eq:source-one-in-three}
	\end{equation}
	be the resulting cubic planar monotone Exact 1-in-3-SAT instance on $m$ variables $z_1,\ldots,z_m$ and $m$ clauses. Put
	\begin{equation}
	t=2m.
	\label{eq:source-vector-dimension}
	\end{equation}
	The reduction produces an MSS instance over a finite extension field~$\Kfield$.  Its design has two logically separate parts:
	\begin{enumerate}
		\item the first moment, after applying a quotient map, certifies a satisfying Boolean assignment;
		\item a universal auxiliary pool realizes all remaining moment residuals without affecting that quotient equation.
	\end{enumerate}
	We now describe every part of the output explicitly.
	
	\paragraph{Step 1: encode the Boolean choices.}
	For each $(i,b)\in[m]\times\{0,1\}$, form the vector
	$v_{i,b}\in\F_q^t$ from \eqref{eq:boolean-choice-vector}.  Its variable part records the choice $z_i=b$, while its clause part records which clause occurrences become true under that choice.  Let
	\begin{equation}
	T=(\one_m;\one_m)\in\F_q^t.
	\label{eq:reduction-vector-target}
	\end{equation}
	By Proposition~\ref{prop:one-in-three-vector-sum}, a subset of the vectors $v_{i,b}$ sums to~$T$ if and only if it selects exactly one value for every variable and that assignment makes exactly one variable true in every clause.
	
	\paragraph{Step 2: choose the moment parameter and the base field.}
	Set $d=t^{\nu}\ge2$ and $h=Cd$, using the constants fixed in preprocessing.  The parameter~$d$ is large enough for \Cref{thm:uniform-completion}.  The auxiliary number~$h$ is used only in the completeness proof through that theorem; it is not part of the output instance.  Put
	\begin{equation}
	Q_0
	=
	\max\left\{
	1+\left\lceil\frac{d(d+1)}2\right\rceil,
	\left\lceil d^{2+\rho}\right\rceil,
	h+1
	\right\}.
	\label{eq:prime-search-start}
	\end{equation}
	Since $\rho$ is rational, the ceiling of $d^{2+\rho}$ can be computed exactly by integer powers and integer-root comparisons.
	Choose a prime~$q$ satisfying
	\begin{equation}
	Q_0\le q\le2Q_0.
	\label{eq:q-bertrand}
	\end{equation}
	Bertrand's postulate guarantees existence. Test each integer in this interval by trial division up to its square root. This uses $O(Q_0^{3/2})$ integer divisions on numbers of polynomial bit length. Since $Q_0$ is polynomial in the source size, the scan takes deterministic polynomial time.  The inequalities in \eqref{eq:prime-search-start} ensure
	\begin{equation}
	q>h>4,
	\qquad
	q>\frac{d(d+1)}2,
	\qquad
	q\ge d^{2+\rho}.
	\label{eq:q-required-properties}
	\end{equation}
	Together with primality of~$q$ and the lower bound on~$d$ ensured by preprocessing, these are the hypotheses needed by the Boolean encoding, the deterministic power condenser, and the elementary uniform completion theorem.
	
	\paragraph{Step 3: construct the extension field and the quotient map.}
	Using the deterministic procedure in \Cref{subsec:explicit-extension-field}, construct
	\begin{equation}
	\Kfield=\F_{q^{2t+1}}
	\label{eq:reduction-extension-field}
	\end{equation}
	with an explicit power-basis representation.  Inside~$\Kfield$, define
	\begin{equation}
	H
	=
	\Span_{\F_q}
	\{1,\theta,\ldots,\theta^t\}
	\label{eq:reduction-H}
	\end{equation}
	and the quotient map
	\begin{equation}
	\pi:\Kfield\longrightarrow\F_q^t
	\label{eq:reduction-quotient}
	\end{equation}
	from \eqref{eq:quotient-isomorphism}.  It satisfies
	\begin{equation}
	\ker\pi=H,
	\qquad
	\Span_{\F_q}(H\cdot H)=\Kfield.
	\label{eq:reduction-quotient-properties}
	\end{equation}
	The first identity will protect soundness: every auxiliary element is placed in~$H$ and therefore vanishes under~$\pi$.  The second identity is the algebraic source of enough higher powers to complete arbitrary residual moments.
	
	\paragraph{Step 4: lift the Boolean vectors to Boolean-choice elements.}
	Let
	\begin{equation}
	\sigma:\F_q^t\longrightarrow\Kfield
	\end{equation}
	be the explicit linear section of~$\pi$ from \Cref{subsec:explicit-extension-field}.  For every Boolean choice $(i,b)$, set
	\begin{equation}
	\ell_{i,b}
	=
	\sigma(v_{i,b}).
	\label{eq:reduction-choice-lift}
	\end{equation}
	Because $\sigma$ is injective and the vectors $v_{i,b}$ are pairwise distinct, all $2m$ Boolean-choice elements $\ell_{i,b}$ are pairwise distinct.  By construction,
	\begin{equation}
	\pi(\ell_{i,b})=v_{i,b}.
	\label{eq:reduction-choice-projection}
	\end{equation}
	Moreover, $v_{i,b}\ne0$, so every Boolean-choice element lies outside~$H$.
	
	\paragraph{Step 5: construct the auxiliary completion pool.}
	Use Lemma~\ref{lem:deterministic-directions} to construct a projectively distinct direction set $B\subset H\setminus\{0\}$ of even cardinality $R\le L+1$.  Form the disjoint union of nonzero base-field lines
	\begin{equation}
	D
	=
	\bigsqcup_{\beta\in B}
	\beta\F_q^\times.
	\label{eq:reduction-auxiliary-pool}
	\end{equation}
	Every element of~$D$ belongs to~$H$, and hence every subset $Y\subseteq D$ satisfies
	\begin{equation}
	\pi\!\left(\sum_{y\in Y}y\right)=0.
	\label{eq:reduction-pool-invisible}
	\end{equation}
	On the other hand, \Cref{thm:balanced-completion} guarantees that for every residual tuple
	\begin{equation}
	r_1\in H,
	\qquad
	r_2,\ldots,r_d\in\Kfield,
	\label{eq:reduction-residual-tuple}
	\end{equation}
	there is a subset $Y\subseteq D$ of the fixed size $|D|/2$ satisfying
	\begin{equation}
	\sum_{y\in Y}y^j=r_j,
	\qquad
	1\le j\le d.
	\label{eq:reduction-residual-completion}
	\end{equation}
	The reduction constructs the pool~$D$, but it does not need to find~$Y$.  In a YES instance, the appropriate set~$Y$ is part of the NP witness.
	
	\paragraph{Step 6: define the MSS instance.}
	The MSS universe is the disjoint union
	\begin{equation}
	A
	=
	\{\ell_{i,b}:i\in[m],\ b\in\{0,1\}\}
	\sqcup D.
	\label{eq:MSS-universe}
	\end{equation}
	It is disjoint because every Boolean-choice element lies outside~$H$, whereas $D\subset H$.  The Boolean-choice elements are pairwise distinct because $\sigma$ is injective, and the auxiliary lines are pairwise disjoint because the directions in~$B$ are projectively distinct.  In particular, $0\notin A$.
	
	Use the explicit section to choose the first-moment target
	\begin{equation}
	M_1=\sigma(T),
	\qquad
	\pi(M_1)=T:=(\one_m;\one_m)\in\F_q^t,
	\label{eq:M1-lift}
	\end{equation}
	and set
	\begin{equation}
	M_2=\cdots=M_d=0.
	\label{eq:higher-targets-zero}
	\end{equation}
	Finally, define
	\begin{equation}
	n=|A|+1=2m+|D|+1.
	\label{eq:output-code-length}
	\end{equation}
	Thus $n$ is already the block length of the Reed--Solomon instance obtained in the next reduction.  Require an MSS subset of size
	\begin{equation}
	k
	=
	m+\frac{|D|}{2}
	=
	\frac{n-1}{2}.
	\label{eq:n-and-k}
	\end{equation}
	The quantity $|D|=|B|(q-1)$ is even because $|B|$ is even and $q$ is odd.
	
	The intended witness now has a transparent form: it contains the $m$ Boolean-choice elements specified by a satisfying assignment and exactly half of the auxiliary pool.  Projecting its first-moment equation by~$\pi$ discards the auxiliary half and recovers the vector equation of \Cref{prop:one-in-three-vector-sum}.  Conversely, once a satisfying assignment has been selected, its first-moment residual lies in~$H$, while its higher-moment residuals are arbitrary elements of~$\Kfield$; \Cref{thm:balanced-completion} supplies the required auxiliary half.  The next subsection formalizes these two directions.
	
	\subsection{Correctness}
	
	\begin{theorem}[Reduction correctness]
		\label{thm:reduction-correctness}
		The formula~$\varphi$ is satisfiable if and only if the MSS instance
		\begin{equation}
		(A,k,M_1,\ldots,M_d)
		\label{eq:constructed-instance}
		\end{equation}
		has a solution.
	\end{theorem}
	
	\begin{proof}[Completeness]
		Suppose~$\varphi$ has an Exact 1-in-3 satisfying assignment.  Select the corresponding Boolean-choice lift for each variable, and call the resulting $m$-element set~$L_0$.  By \eqref{eq:vector-subset-sum} and \eqref{eq:reduction-choice-projection},
		\begin{equation}
		\pi\!\left(
		M_1-\sum_{\ell\in L_0}\ell
		\right)=0.
		\end{equation}
		Therefore the first-moment residual
		\begin{equation}
		r_1=M_1-\sum_{\ell\in L_0}\ell
		\label{eq:first-residual}
		\end{equation}
		lies in~$H$.  For $2\le j\le d$, define
		\begin{equation}
		r_j
		=M_j-\sum_{\ell\in L_0}\ell^j
		=-\sum_{\ell\in L_0}\ell^j
		\in\Kfield.
		\label{eq:higher-residual}
		\end{equation}
		The second equality uses $M_j=0$ from
		\eqref{eq:higher-targets-zero}.
		By \Cref{thm:balanced-completion}, there is a subset
		$Y\subseteq D$ of size $|D|/2$ whose moments are
		$r_1,\ldots,r_d$.  Then
		\begin{equation}
		S=L_0\sqcup Y
		\end{equation}
		has size~$k$ and realizes the targets
		$M_1,\ldots,M_d$.
	\end{proof}
	
	\begin{proof}[Soundness]
		Suppose $S\subseteq A$ is any solution.  Apply the quotient map to its first-moment equation.  Every selected auxiliary element disappears by \eqref{eq:quotient-zero}, so
		\begin{equation}
		\sum_{\ell_{i,b}\in S}
		v_{i,b}=T.
		\label{eq:soundness-quotient}
		\end{equation}
		By Proposition~\ref{prop:one-in-three-vector-sum}, the selected Boolean-choice lifts choose exactly one value for every variable and make exactly one variable true in every clause.  Hence~$\varphi$ is satisfiable.
		
		No property of the selected auxiliary subset is used.  In particular, soundness does not require it to contain $h$ elements from positive lines or complements of $h$-sets from negative lines.  Those structures are used only to prove completeness.
	\end{proof}
	
	\subsection{Size and running time}\label{subsec:size-running-time}
	
	The extension degree is $L=2t+1=4m+1$.  By
	Lemma~\ref{lem:deterministic-directions},
	\begin{equation}
	R=|B|=L+1=2t+2.
	\label{eq:R-linear-size}
	\end{equation}
	Using $q=O(d^{2+\rho})$ and \eqref{eq:output-code-length}, the eventual Reed--Solomon block length satisfies
	\begin{equation}
	n
	=2m+R(q-1)+1
	=O\!\left(t\,d^{2+\rho}\right).
	\label{eq:n-deterministic}
	\end{equation}
	
	For the fixed rational $0<\alpha<1/2$, choose a positive rational $\rho$ such that
	\begin{equation}
	\alpha<\frac{1}{2+\rho}.
	\end{equation}
	Set $d=t^{\nu}$ for a sufficiently large fixed integer~$\nu$.  From
	\eqref{eq:n-deterministic},
	\begin{equation}
	\liminf_{t\to\infty}
	\frac{\log d}{\log n}
	\ge
	\frac{\nu}{1+\nu(2+\rho)}.
	\label{eq:deterministic-exponent}
	\end{equation}
	The right-hand side approaches $1/(2+\rho)$ as $\nu\to\infty$, so choose~$\nu$ such that this ratio exceeds~$\alpha$, as required in \eqref{eq:effective-reduction-constants}.  The construction then has at least
	$\lfloor n^\alpha\rfloor$ moment equations for all sufficiently large~$t$, uniformly over source formulas of that size.  Likewise, $\lfloor n^\alpha\rfloor\le(n-1)/2$ for all sufficiently large output lengths.  Both thresholds are included in the preprocessing choice of~$t_0$.
	
	Write $\alpha=a/b$ with fixed positive integers $a,b$.  Compute $\lfloor n^\alpha\rfloor$ exactly by binary search for the largest nonnegative integer~$D$ satisfying $D^b\le n^a$.  This uses polynomial-time integer arithmetic.  Retain only the first $\lfloor n^\alpha\rfloor$ equations and denote this final number of moments by~$d$.  Completeness is preserved, and soundness still follows from the first moment alone.  This establishes the required number of moments in \Cref{thm:main-mss}; the alphabet-size assertion is verified below.
	
	By \Cref{subsec:explicit-extension-field}, the field, its power basis, and all required arithmetic are constructed deterministically in polynomial time.  Equation~\eqref{eq:field-element-size} also shows that the total bit length of the output is polynomial in the input length.
	
	\paragraph{Alphabet size compared with the evaluation set.}
	Write $d_{\mathrm{pre}}=t^{\nu}$ for the number of moments before truncation. By the prime selection in \eqref{eq:prime-search-start}--\eqref{eq:q-bertrand},
	\begin{equation}
	q=\Theta(d_{\mathrm{pre}}^{2+\rho})
	=\Theta(t^{\nu(2+\rho)}).
	\end{equation}
	Since $t=2m$, the direction construction has exactly $R=2t+2=4m+2$ directions. Thus
	\begin{equation}
	n=2m+(4m+2)(q-1)+1
	=\Theta(tq)
	=\Theta\!\left(t^{1+\nu(2+\rho)}\right).
	\label{eq:exact-length-field-comparison}
	\end{equation}
	Put
	\begin{equation}
	\eta=\frac{1}{1+\nu(2+\rho)}\in(0,1).
	\label{eq:alphabet-exponent}
	\end{equation}
	The alphabet size $Q=|\Kfield|$ therefore satisfies
	\begin{equation}
	Q=q^{2t+1},
	\qquad
	\log_2 Q
	=(2t+1)\log_2q
	=\Theta(t\log t)
	=\Theta(n^\eta\log n)
	=o(n).
	\label{eq:subexponential-alphabet}
	\end{equation}
	Consequently $Q=2^{\Theta(n^\eta\log n)}=2^{o(n)}$, as claimed in \Cref{thm:main-mss,thm:main-bdd}. Truncating the moment list changes none of $A,n,q,L$, and the reduction to Reed--Solomon decoding uses the same field. The alphabet remains superpolynomial in~$n$, and the evaluation set occupies a vanishing fraction of the field. Nevertheless, each field element requires only $\Theta(n^\eta\log n)$ bits in the specified representation, consistently with the polynomial output-size bound above.
	
	\section{Proof of the Reed--Solomon BDD theorem}
	\label{sec:rs-consequences}
	
	\begin{proof}[Proof of \Cref{thm:main-bdd}]
		Fix a rational $0<\alpha<1/2$.  By \Cref{thm:main-mss}, there is a deterministic polynomial-time reduction from Exact 1-in-3-SAT to an MSS instance over an explicitly represented extension field in which the input set has size $n-1$, the required subset size is
		\begin{equation}
		k=\frac{n-1}{2},
		\end{equation}
		and the number of moments is
		\begin{equation}
		d=\lfloor n^\alpha\rfloor.
		\end{equation}
		The field characteristic exceeds~$d$, and its size satisfies \eqref{eq:subexponential-alphabet}.
		
		Apply \Cref{prop:mss-to-rs}.  The resulting Reed--Solomon code has block length~$n$ and dimension
		\begin{equation}
		K=k-d+1
		=\frac{n+1}{2}-d.
		\label{eq:code-dimension-final}
		\end{equation}
		The decoding radius is
		\begin{align}
		(n-K)-d
		&=n-\left(\frac{n+1}{2}-d\right)-d \\
		&=\frac{n-1}{2}.
		\label{eq:code-radius-final}
		\end{align}
		The covering radius is
		\begin{equation}
		n-K=\frac{n-1}{2}+d,
		\end{equation}
		so the requested radius is exactly~$d$ below it.  Equivalently, one asks whether a polynomial of degree less than~$K$ agrees with the received word in at least
		\begin{equation}
		K+d=\frac{n+1}{2}
		\label{eq:agreement-threshold-final}
		\end{equation}
		positions.
		
		By \Cref{prop:mss-to-rs}, the received word is within this radius if and only if the MSS instance is satisfiable.  This proves NP-hardness.  Membership in NP follows because a polynomial of degree less than~$K$ is a polynomial-size certificate whose evaluation and Hamming distance can be checked in deterministic polynomial time.  Finally,
		\begin{equation}
		\frac{K}{n}
		=\frac{(n+1)/2-d}{n}
		\longrightarrow\frac12,
		\end{equation}
		since $d=o(n)$.
	\end{proof}
	
	\section{Why the proof needs extension fields}
	\label{sec:prime-field}
	
	The role of the extension field should be separated from the role of the prime subfield.
	
	\paragraph{What the prime field does.}
	The completion lemma is a statement over~$\F_q$.  Every direction line
	$D_\beta=\beta\F_q^\times$ is parametrized by base-field scalars, and
	the point counts produce the desired subsets of~$\F_q^\times$.
	
	\paragraph{What the extension field does.}
	The extension supplies a proper additive subspace~$H$ satisfying two incompatible-looking requirements:
	\begin{equation}
	\dim(\Kfield/H)=t>0,
	\qquad
	\Span(H\cdot H)=\Kfield.
	\label{eq:two-H-requirements}
	\end{equation}
	The quotient stores the Boolean subset-sum instance.  The product-spanning property lets higher powers of auxiliary directions correct arbitrary residuals in the full field.
	
	\paragraph{Why the same quotient construction cannot work over a prime field.}
	View a prime field $\F_p$ as a vector space over itself.  Its only linear subspaces are
	\begin{equation}
	H=0
	\qquad\text{and}\qquad
	H=\F_p.
	\end{equation}
	If $H=0$, there is no nonzero auxiliary pool.  If $H=\F_p$, then the quotient is zero and the first moment carries no Boolean information.  Thus no $\F_p$-linear quotient construction can simultaneously hide all auxiliaries and retain a nontrivial soundness equation.
	
	One might instead regard $\F_p$ as a vector space over a proper subfield, but a prime field has no proper subfields.  This is the exact point at which the argument fails.  It does not prove that polynomial-moment hardness over prime fields is false; it shows only that a qualitatively different soundness mechanism is required.  The magnitude-bimodality method of~\cite{GandikotaGhaziGrigorescu2018} is one such prime-field mechanism, but it currently incurs an exponential-in-$d$ gadget size.
	
	\section{Comparison, open problems, and perspective}
	\label{sec:perspective}
	
	\subsection{Coding-theoretic comparison with earlier results}
	
	For an $[n,K]$ Reed--Solomon code, write
	\begin{equation}
	d=(n-K)-\tau
	\end{equation}
	for the gap between the covering radius and the decoding radius~$\tau$.  The relevant algorithmic and hardness benchmarks are summarized below.  The full-length discrete-logarithm reductions concern search decoding; the NP-completeness results concern the threshold decision problem for constructed evaluation sets.
	
	\begingroup
	\hbadness=10000
	\begin{center}
		\small
		\renewcommand{\arraystretch}{1.12}
		\begin{tabularx}{\textwidth}{@{}
				>{\RaggedRight\arraybackslash}p{0.18\textwidth}
				>{\RaggedRight\arraybackslash}p{0.23\textwidth}
				>{\RaggedRight\arraybackslash}p{0.17\textwidth}
				Y@{}}
			\toprule
			Result & Decoding regime & Field model & Significance \\
			\midrule
			Guruswami--Sudan~\cite{GuruswamiSudan1999}
			& Polynomial-time list decoding up to approximately $n-\sqrt{n(K-1)}$
			& Arbitrary finite fields admitting the evaluation set
			& Standard Johnson-radius algorithmic benchmark. \\
			\addlinespace
			Brakensiek et al.~\cite{Br+26}
			& Relative radius $1-\eps$ for sufficiently small fixed $\eps$ and $R\le(1-\theta)\eps$
			& Prime fields of size $O(n)$; arbitrary evaluation sets
			& Polynomial-time list decoding beyond the Johnson radius at low constant rates; fixed $\theta\in(0,1)$. \\
			\addlinespace
			Guruswami--Vardy; Cheng--Murray~\cite{GuruswamiVardy2005,ChengMurray2007}
			& $d=1$
			& Constructed evaluation sets in large fields
			& NP-completeness of non-deep-hole recognition; the complementary deep-hole problem is coNP-complete. \\
			\addlinespace
			Cheng--Wan~\cite{ChengWan2007,ChengWan2010}
			& $d=h$; additive gaps growing as a fixed positive power of the length in suitable families, including positive rates
			& Standard primitive and extended full-length codes
			& Randomized reductions from extension-field discrete logarithms to search decoding. \\
			\addlinespace
			GGG (2018)~\cite{GandikotaGhaziGrigorescu2018}
			& $1\le d\le c\log n/\log\log n$ for an absolute $c>0$
			& Prime fields of size $2^{\poly(n)}$; constructed sets; constant rate
			& Deterministic NP-hardness with an asymptotically growing gap below the covering radius. \\
			\addlinespace
			Present result
			& $d=\Theta(n^\alpha)$ for every fixed rational $0<\alpha<1/2$
			& Extension fields of size $2^{o(n)}$; constructed sets; rate tending to $1/2$
			& Deterministic NP-hardness at an additive gap of order $n^\alpha$ below the covering radius. \\
			\bottomrule
		\end{tabularx}
	\end{center}
	\endgroup
	
	The construction combines an additive gap of order $n^\alpha$ with an explicitly subexponential alphabet, as quantified in \eqref{eq:subexponential-alphabet}. For comparison, the earlier growing-gap theorem~\cite{GandikotaGhaziGrigorescu2018} is stated over fields of size $2^{\poly(n)}$. The present bound concerns the full extension field $\Kfield$, not merely its prime subfield~$\F_q$.
	
	The advance is additive rather than multiplicative: the hard decoding radius remains $n-K-o(n)$, so the result does not approach the Johnson radius in relative terms.  Within the deterministic NP-hardness framework, it replaces the previous polylogarithmic gap below the covering radius by an additive gap of order $n^\alpha$ over extension fields.  The evaluation set remains part of the constructed instance; NP-hardness for standard primitive or extended full-length Reed--Solomon codes is not established here.
	
	Moments subset sum is the intermediate algebraic problem that makes this improvement possible.  The earlier reduction requires an explicit, soundness-compatible PTE bundle for every Boolean alternative.  Our reduction constructs one universal pool, delegates the arithmetic completion to the NP witness, and makes every unintended auxiliary subset harmless through an extension-field quotient.  Related work on moments subset sum for structured evaluation sets develops algorithms, point counts, and all-target existence theorems rather than hardness for arbitrary input sets~\cite{LaiMarinoRobinsonWan2020,GottigPerezPrivitelli2024}.
	
	\subsection{The deterministic power condenser}
	
	The linear-size direction set is deterministic.  The essential ingredients are special to the maps
	$x\mapsto x^j$: pure powers span the extension by polarization, affine restrictions admit the exact identity test in
	Lemma~\ref{lem:affine-power-identity-test}, and coordinate fixing encounters at most
	\begin{equation}
	\sum_{j=1}^{d}j
	\end{equation}
	forbidden base-field values at each stage.  The field-size condition
	$q>d(d+1)/2$ is already implied by the point-counting regime
	$q\ge d^{2+\rho}$.
	
	This construction may be useful independently as a deterministic simultaneous rank condenser for power maps in finite-dimensional field algebras.  It is important, however, that the proof exploits the product-spanning property
	$\Span(H\cdot H)=\Kfield$ and the fact that the characteristic exceeds all relevant degrees.
	
	\subsection{Improving the completion field size}
	
	The elementary completion theorem used in the reduction naturally leads to
	$q\ge d^{2+\rho}$ for an arbitrary fixed $\rho>0$.  The stronger higher-dimensional estimate in~\cite{Wan2026} improves the admissible constant multiplying~$d$ in the completion size, but it does not change this field-size exponent in the present reduction.  A stronger sieve or additional cancellation might approach the information-theoretic scale.  Universal completion for all $q^d$ targets requires at least $q^d$ candidate subsets, while a field of size~$q$ has only $2^q$ subsets.  Hence
	\begin{equation}
	2^q\ge q^d,
	\qquad\text{so}\qquad
	q=\Omega(d\log d).
	\label{eq:counting-lower-bound}
	\end{equation}
	This leaves $q=d^{1+o(1)}$ as an optimistic field-size target, subject to further constraints.  In particular, when every completion has the fixed cardinality $h=\lceil Cd\rceil$, the sharper necessary condition is
	\begin{equation}
	\binom{q-1}{h}\ge q^d.
	\label{eq:fixed-cardinality-counting}
	\end{equation}
	For fixed~$C$ and $q=d^{1+o(1)}$, however,
	\begin{equation}
	\log\binom{q-1}{h}
	\le h\log\!\left(\frac{e(q-1)}{h}\right)
	=o(d\log d),
	\end{equation}
	whereas $d\log q=(1+o(1))d\log d$.  Thus reaching this field-size scale would require abandoning the fixed bound $h=O(d)$.  The complement construction in \Cref{thm:balanced-completion} can accommodate a different common cardinality~$h$, provided the required all-target completion theorem holds.
	
	There is also an independent obstruction in the direction construction: Lemma~\ref{lem:deterministic-directions} requires
	$q>d(d+1)/2$.  Consequently, improving the completion theorem alone does not remove the quadratic field-size requirement.  A route toward hardness for every fixed rational $0<\alpha<1$ would require both a completion theorem over smaller base fields with suitable cardinalities and a deterministic direction construction valid in that regime.  For example, if both were available with $q=d^{1+o(1)}$ and $|B|=O(t)$, the resulting bound $n=O(t\,d^{1+o(1)})$ would support the same parameter argument for such exponents.  Neither improvement is proved here.
	
	\subsection{An alternative soundness construction over prime fields}
	
	The most coding-theoretically direct open problem is to recover polynomial-moment hardness over prime fields.  Any replacement must allow all auxiliary subsets to be sound while retaining enough multiplicative freedom to correct higher moments.  Possible directions include:
	\begin{itemize}
		\item a code quotient spread across several prime-field coordinates;
		\item a ring or module quotient construction followed by an embedding into a field;
		\item a positional construction with algebraic, rather than magnitude-only, carry prevention;
		\item a structured family of signed or matrix-valued completion measurements.
	\end{itemize}
	None of these is supplied by the present proof.

	%------------------------------------------------
	% Author addresses
	%------------------------------------------------
	
	\bigskip
	
	\noindent
	\textsc{Daqing Wan and Jun Zhang}\\
	Center for Discrete Mathematics,\\ College of Mathematics and Statistics,\\
	Chongqing University, Chongqing 401331, China\\
	Email: \texttt{dwan@uci.edu.cn, jun\_zhang@cqu.edu.cn}

\end{document}